\documentclass[10pt]{article}

\usepackage{epsfig}
\usepackage{graphicx}
\usepackage[latin1]{inputenc}
\usepackage{amsmath}
\usepackage{amsthm}
\usepackage{amsfonts}
\usepackage{amssymb}
\usepackage{amstext}
\usepackage{amsgen}
\usepackage{amsbsy}
\usepackage{amsopn}
\usepackage{hyperref}
\usepackage{mathrsfs,color}
\usepackage{a4wide}

\numberwithin{equation}{section}

\theoremstyle{plain}
\newtheorem{theorem}{Theorem}[section]
\newtheorem{corollary}[theorem]{Corollary}
\newtheorem{proposition}[theorem]{Proposition}
\newtheorem{lemma}[theorem]{Lemma}

\theoremstyle{definition}
\newtheorem{definition}[theorem]{Definition}
\newtheorem*{examples}{Examples}

\theoremstyle{remark}
\newtheorem{remark}[theorem]{Remark}

\newcommand{\supp}{\mathop{\mathrm{supp}}}
\newcommand\C{\ensuremath{\mathbb{C}}}
\newcommand\N{\ensuremath{\mathbb{N}}}
\newcommand\R{\ensuremath{\mathbb{R}}}

\newcommand\eps{\varepsilon}

\newcommand{\SobSDt}[2]{^{\nabla}\|#1\|_{S_\tau,\,\eps}^{#2}}
\newcommand{\SobSDz}[2]{^{\nabla}\|#1\|_{S_\zeta,\,\eps}^{#2}}

\newcommand{\SobSdt}[2]{^{\partial}\|#1\|_{S_\tau,\,\eps\emph{}}^{#2}}
\newcommand{\SobSdz}[2]{^{\partial}\|#1\|_{S_\zeta,\,\eps\emph{}}^{#2}}
\newcommand{\SobODt}[2]{^{\nabla}\|#1\|_{\Omega_\tau,\,\eps}^{#2}}
\newcommand{\SobOdt}[2]{^{\partial}\|#1\|_{\Omega_\tau,\,\eps\emph{}}^{#2}}

\newcommand{\Cc}{\ensuremath{\mathcal{C}}}
\newcommand{\Cinfty}{\ensuremath{{\Cc}^\infty}}
\newcommand{\esm}{\ensuremath{\mathcal{E}_M}}

\newcommand{\ns}{\ensuremath{\mathcal{N}}}

\newcommand{\gs}{\ensuremath{\mathcal{G}}}
\newcommand{\al}{\alpha}
\newcommand{\pa}{\partial}

\newcommand\eep[1]{\ensuremath{e^{\eps}_{#1}}}

\newcommand\stack[2]{\genfrac{}{}{0pt}{}{#1}{#2}}

\newcommand\bsig{\ensuremath{\boldsymbol{\sigma}}}
\newcommand\bsighat{\ensuremath{\widehat{{\boldsymbol{\sigma}}}}}
\newcommand\bsigepshat{\ensuremath{\widehat{{\boldsymbol{\sigma}}_{\eps}}}}

\newcommand\bxi{\ensuremath{{\boldsymbol{\xi}}}}
\newcommand\bxihat{\ensuremath{\widehat{\boldsymbol{\xi}}}}
\newcommand\bxieps{\ensuremath{{\boldsymbol{\xi}}_{\eps}}}
\newcommand\bxiepshat{\ensuremath{\widehat{{\boldsymbol{\xi}}_{\eps}}}}

\begin{document}
\title{The wave equation on singular space-times}
\author{James~D.E.\ Grant, Eberhard Mayerhofer, Roland Steinbauer\footnote{Faculty of Mathematics, University of Vienna, Nordbergstrasse 15, 1090 Vienna, Austria, email: \{james.grant, eberhard.mayerhofer, roland.steinbauer\}@univie.ac.at; supported by FWF research grants P16742-N04 and Y237-N13}}
\date{19 February, 2008}
\maketitle
\begin{abstract}
We prove local unique solvability of the wave equation for a large
class of weakly singular, locally bounded space-time metrics in a
suitable space of generalised functions.



\noindent {\em PACS 2008:} 04.20.Dw, 04.20.Ex,  02.30.Jr,  02.30.Sa
\end{abstract}
\thispagestyle{empty}

\section{Introduction}

The notion of a singularity in general relativity significantly
differs from that in other field theories. In the absence of a
background metric, one has to detect the presence of singularities
by showing that the space-time is \lq\lq incomplete\rq\rq\ in some
sense. In the standard approach to singularities (see, e.g., Hawking
and Ellis~\cite[Ch.\ 8]{HE}), a singularity is regarded as an
obstruction to extending geodesics. However, this definition does
not correspond very closely to ones physical intuition and
classifies many space-times that have been used to model physically
reasonable scenarios as being \lq\lq singular\rq\rq. Such \lq\lq
weakly singular\rq\rq\ space-times have long been used to describe,
for example, impulsive gravitational waves, shell-crossing
singularities and thin cosmic strings. Typically these space-times
admit a metric that is locally bounded but its differentiability is
below $\Cc^{1, 1}$ (i.e., the first derivative locally Lipschitz)
--- the largest differentiability class where standard differential
geometric properties, such as existence and uniqueness of geodesics,
remain valid. For a recent review on the use of metrics of low
regularity in general relativity, see~\cite{SV}.

This set of problems has stimulated considerations of whether
physical objects would be subjected to unbounded tidal forces on
approaching the singularity and was formulated mathematically in
terms of strong curvature conditions. Unfortunately, it is hard to
model the behaviour of real physical objects in a strong
gravitational field. This led Clarke~\cite{clarkeINDIEN} to suggest
that one consider the behaviour of physical fields (for which one
has a precise mathematical description) near the singularity
instead. According to this philosophy of \lq\lq generalised
hyperbolicity\rq\rq\ one should regard singularities as obstructions
to the Cauchy development of these fields rather than as an
obstruction to the extension of geodesics.

However, the weak singularities mentioned above are obstructions if
one formulates the Cauchy problem for the wave equation in the
standard theory of distributions. More precisely, there is no
generally valid distributional solution concept for the wave
equation on a space-time with a non-smooth metric. The equation,
although linear, involves coefficients of low regularity that cannot
be multiplied with the distributional solution.

To resolve this problem in the case of shell-crossing singularities,
Clarke~\cite{Clarke} introduced a specific weak solution concept
(called \emph{$\Box$-global hyperbolicity}) to prove unique
solvability of the wave equation, hence showing that these
space-times, indeed, satisfy the conditions of generalised
hyperbolicity. On the other hand, Vickers and Wilson~\cite{VW} used
the setting of Colombeau algebras~\cite{Colombeau, C} to arrive at a
valid formulation of the Cauchy problem for the wave equation on
conical space-times (modelling a thin cosmic string) and showed the
existence and uniqueness of solutions in a suitable algebra,
$\mathcal{G}$, of generalised functions. Hence they showed that
conical space-times are generalised hyperbolic or, more precisely,
$\mathcal{G}$-hyperbolic. Vickers and Wilson also showed that their
unique generalised solution corresponds to the \lq\lq
forbidden\rq\rq\ distributional solution expected on physical
grounds (via the concept of association --- see Section~\ref{geomth}
below). Their key tool is a refinement of the energy estimates for
hyperbolic PDEs (see, e.g.,~\cite[Sec.\ 7.4]{HE}, \cite[Sec.\
4.4]{ClarkeBook}), which makes them applicable in the new situation.

In this paper, we generalise this method to a much wider class of
weakly singular space-times and prove $\mathcal{G}$-hyperbolicity
for this class. Since our approach is based on regularisation of the
singular metric by sequences of smooth ones, we must put
restrictions on the growth of the sequence with respect to the
regularisation parameter $\eps$. Essentially we shall assume (see
Section~\ref{definitions2}) asymptotic local uniform boundedness
with respect to $\eps$. Recall that the space-times of interest here
typically possess a locally bounded metric. In particular, our class
includes impulsive $pp$-waves (in the Rosen form), expanding
spherical impulsive waves, and conical space-times (thereby
generalising the results of Vickers and Wilson~\cite{VW}).

\vskip .2cm
This work is organised in the following way. In
Section~\ref{prereq} we fix our notation, recall some facts on the
geometric theory of generalised functions and define our class of
weakly singular space-times (Section~\ref{definitions2}). We state
our main result in Theorem~\ref{mainthm} of
Section~\ref{mainresult}: given a point $p$ in a weakly singular
space-time, there exists a neighbourhood, $V$, of $p$ such that the
initial value problem for the wave equation admits a unique solution
in $\mathcal{G}(V)$. The proof is split into several steps:
(generalised) higher order energy integrals are introduced and
proved to be equivalent to suitable Sobolev norms in
Section~\ref{eint}. The energy estimates are provided in
Section~\ref{partb}, while some auxiliary estimates are proved in
Section~\ref{auxest}. Finally these results are collected to provide
the proof of the main theorem in Section~\ref{proof}. We end with
some concluding remarks.

\section{Prerequisites}
\label{prereq}

In this section, we give a precise definition of the class of weakly
singular metrics that we are going to consider in the sequel. Prior
to that, and for the convenience of the reader, we give a brief
summary of the geometric theory of generalised functions in the
sense of Colombeau. Our main reference for the latter is~\cite[Sec.\
3.2]{gkos} and we adopt most notational conventions from there. For
an overview of the use of these constructions in general relativity,
we refer to~\cite{SV}.

\subsection{Geometric theory of generalised functions}
\label{geomth}

The basic idea of Colombeau's approach to generalised
functions~\cite{Colombeau, C} is regularisation by sequences (nets)
of smooth functions and the use of asymptotic estimates in terms of
a regularisation parameter $\eps$. Let $M$ be a separable, smooth,
orientable, Hausdorff manifold of dimension $n$, and let ${\mathfrak
X}(M)$ denote the space of smooth vector fields on $M$. Let
$(u_\eps)_{\eps \in (0, 1]}$ with $u_\eps \in \Cc^\infty(M)$ for all
$\eps$. The (special) algebra of generalised functions on $M$ is
defined as the quotient $\gs(M) := \esm(M)/\ns(M)$ of the moderate
nets modulo the negligible nets, where the respective notions are
defined by the following asymptotic estimates:
\begin{eqnarray*}
\esm(M) &:=& \{ (u_\eps)_\eps:\ \forall K \subset\subset M, \,
\forall k \in \N_0\, \exists N \in \N \
\\
&&\forall \boldsymbol{\eta}_1, \dots, \boldsymbol{\eta}_k \in
\mathfrak{X}(M):\ \sup_{p \in K} | \mathscr{L}_{\boldsymbol{\eta}_1}
\dots \mathscr{L}_{\boldsymbol{\eta}_k} \, u_\eps(p) | =
O(\eps^{-N}) \},
\\
\ns(M) &:=& \{ (u_\eps)_\eps:\ \forall K \subset\subset M, \,
\forall k, q \in \N_0
\\
&&\forall \boldsymbol{\eta}_1, \dots, \boldsymbol{\eta}_k \in
\mathfrak{X}(M):\ \sup_{p \in K} | \mathscr{L}_{\boldsymbol{\eta}_1}
\dots \mathscr{L}_{\boldsymbol{\eta}_k}\, u_\eps(p) | = O(\eps^{q}))
\}.
\end{eqnarray*}

Elements of $\gs(M)$ are denoted by $u = [(u_\eps)_\eps] =
(u_\eps)_\eps + \ns(M)$. With component-wise operations, $\gs(M)$ is
a fine sheaf of differential algebras with respect to the Lie
derivative with respect to classical vector fields defined by
$\mathscr{L}_{\boldsymbol{\eta}} u :=
[(\mathscr{L}_{\boldsymbol{\eta}} u_\eps)_\eps]$. The spaces of
moderate resp.\ negligible sequences and hence the algebra itself
may be characterised locally, i.e., $u \in \gs(M)$ iff $u \circ
\psi_\al \in \gs(\psi_\al(V_\al))$ for all charts $(V_\al,
\psi_\al)$, where, on the open set $\psi_\al(V_\al) \subset \R^n$,
Lie derivatives are replaced by partial derivatives in the
respective estimates. Smooth functions are embedded into $\gs$
simply by the \lq\lq constant\rq\rq\ embedding $\sigma$, i.e.,
$\sigma(f) := [(f)_\eps]$. On open sets of $\R^n$, compactly
supported distributions are embedded into $\gs$ via convolution with
a mollifier $\rho \in \mathscr{S}(\R^n)$ with unit integral
satisfying $\int \rho(x) x^\alpha dx = 0$ for all $|\alpha| \geq 1$;
more precisely setting $\rho_\eps(x) = (1/\eps^n) \rho(x/\eps)$, we
have $\iota(w) := [(w * \rho_\eps)_\eps]$. In the case where
$\supp(w)$ is non-compact, one uses a sheaf-theoretical construction
which can be lifted to the manifold using a partition of unity.
{}From the explicit formula, it is clear that the embedding commutes
with partial differentiation. This embedding, however, is not
canonical since it depends on the mollifier as well as the partition
of unity. A canonical embedding of $\mathscr{D}^{\prime}$
\emph{is\/} provided by the so-called full version of the
construction (see~\cite{GlobTh}, resp.~\cite{GlobTh2} for the tensor
case). However, since we will model our weakly singular metrics in
generalised functions from the start (see Sec.~\ref{definitions1}
and the discussion at the end of Section~\ref{definitions2}
below) we have chosen to work in the so-called special setting which
is technically more accessible. Note that this is in contrast
to~\cite{VW}.

Inserting $p \in M$ into $u \in \gs(M)$ yields a well-defined
element of the ring of constants (also called generalised numbers)
$\mathcal{K}$ (corresponding to $\mathbb{K} = \R$ resp.\ $\C$),
defined as the set of moderate nets of numbers ($(r_\eps)_\eps \in
\mathbb{K}^{(0, 1]}$ with $|r_\eps| = O(\eps^{-N})$ for some $N$)
modulo negligible nets ($|r_\eps| = O(\eps^{m})$ for each $m$).
Finally, generalised functions on $M$ are characterised by their
generalised point values, i.e., by their values on points in
$\tilde{M}_c$, the space of equivalence classes of compactly
supported nets $(p_\eps)_\eps \in M^{(0, 1]}$ with respect to the
relation $p_\eps \sim p'_\eps :\Leftrightarrow d_h(p_\eps, p'_\eps)
= O(\eps^m)$ for all $m$, where $d_h$ denotes the distance on $M$
induced by any Riemannian metric.

As is evident from the definitions, all estimates are only required
to hold for $\eps$ small enough, that is there exists $\eps_0$ such
that for all $\eps < \eps_0$ the respective statement holds.
However, in order not to unnecessarily complicate our formulations
we will notationally suppress this fact most of the time.

The $\gs(M)$-module of generalised sections in vector bundles ---
especially the space of generalised tensor fields $\gs^r_s(M)$ ---
is defined along the same lines using analogous asymptotic estimates
with respect to the norm induced by any Riemannian metric on the
respective fibers. However, it is more convenient to use the
following algebraic description of generalised tensor fields
\begin{equation}
\label{tensorp} \gs^r_s(M) = \gs(M) \otimes \mathcal{T}^r_s(M)\,,
\end{equation}
where $\mathcal{T}^r_s(M)$ denotes the space of smooth tensor fields
and the tensor product is taken over the module $\Cinfty(M)$. Hence
generalised tensor fields are just given by classical ones with
generalised coefficient functions. Many concepts of classical tensor
analysis carry over to the generalised setting~\cite{KS}, in
particular Lie derivatives with respect to both classical and
generalised vector fields, Lie brackets, exterior algebra, etc.
Moreover, generalised tensor fields may also be viewed as
$\gs(M)$-multilinear maps taking generalised vector and covector
fields to generalised functions, i.e., as $\gs(M)$-modules we have
\begin{equation}
\label{mla} \gs^r_s(M) \cong L_{\gs(M)}(\gs^0_1(M)^r, \gs^1_0(M)^s;
\gs(M)).
\end{equation}

Finally, in light of the Schwartz impossibility result~\cite{Schw1},
the setting introduced above gives a minimal framework within which
tensor fields may be subjected to nonlinear operations, while
maintaining consistency with smooth geometry and allowing an
embedding of the distributional geometry as developed
in~\cite{Marsden, Parker}. Moreover, the interplay between
generalised functions and distributions is most conveniently
formalised in terms of the notion of \emph{association}. A
generalised function $u \in \gs(M)$ is called associated to zero, $u
\approx 0$, if one (hence any) representative $(u_\eps)_\eps$
converges to zero weakly. The equivalence relation $u \approx v
:\Leftrightarrow u - v \approx 0$ gives rise to a linear quotient of
$\gs$ that extends distributional equality. Moreover, we call a
distribution $w \in \mathscr{D}^{\prime}(M)$ the
\emph{distributional shadow\/} or \emph{macroscopic aspect of $u$\/}
and write $u \approx w$ if, for all compactly supported $n$-forms
$\boldsymbol{\nu}$ and one (hence any) representative
$(u_\eps)_\eps$, we have
\[
\lim_{\eps \to 0} \int\limits_M u_\eps \boldsymbol{\nu} = \langle w,
\boldsymbol{\nu} \rangle,
\]
where $\langle \, , \, \rangle$ denotes the distributional action.
By~\eqref{tensorp}, the concept of association extends to
generalised tensor fields in a natural way.

\subsection{Elements of Lorentzian geometry}
\label{definitions1}

A \emph{generalised pseudo-Riemannian metric\/} is defined to be a
symmetric, generalised $(0, 2)$-tensor field $\mathbf{g}$ with a
representative $\mathbf{g}_{\eps}$ that is a smooth
pseudo-Riemannian metric for each $\eps$ such that the determinant
$\det( \mathbf{g} )$ is invertible in the generalised sense. The
latter condition is equivalent to the following notion called
\emph{strictly nonzero on compact sets}: for any representative
$\left( \det( \mathbf{g}_{\eps}) \right)_\eps$ of $\det( \mathbf{g}
)$ we have $\forall K \subset\subset M\ \exists m \in \N:\ \inf_{p
\in K} |\det( \mathbf{g}_{\eps})| \geq \eps^m$. This notion captures
the intuitive idea of a generalised metric as a net of classical
metrics approaching a singular limit in the following precise sense:
$\mathbf{g}$ is a generalised metric iff on every relatively compact
open subset $V$ of $M$ there exists a representative $(
\mathbf{g}_{\eps})_\eps$ of $\mathbf{g}$ such that, for fixed
$\eps$, $\mathbf{g}_{\eps}$ is a classical metric and its
determinant, $\det( \mathbf{g} )$, is invertible in the generalised
sense, i.e., does not go to zero too fast as $\eps \to 0$. Note that
we work exclusively with representatives of generalised metrics that
are classical metrics for each $\eps$. If $\mathbf{g}$ is
Lorentzian, i.e., there exists a representative which is Lorentzian,
we call the pair $(M, \mathbf{g})$ a \emph{generalised space-time}.

A generalised metric induces a $\gs(M)$-linear isomorphism from
$\gs^1_0(M)$ to $\gs^0_1(M)$. The inverse of this isomorphism gives
a well-defined element of $\gs^2_0(M)$ (i.e., independent of the
representative $( \mathbf{g}_{\eps})_\eps$). This is the \lq\lq
inverse metric\rq\rq, which we denote by $\mathbf{g}^{-1}$, with
representative $\left( \mathbf{g}_{\eps}^{-1} \right)_\eps$. The
generalised covariant derivative, as well as the generalised
Riemann-, Ricci- and Einstein tensors, of a generalised metric is
defined by the usual formulae on the level of representatives. For
further details see~\cite{KS1}.

Next, we review the concept of causality in the generalised
framework. Let $\boldsymbol{\xi} \in \mathcal{G}^1_0(M)$ be a
generalised vector field on $M$. Then, by~\eqref{mla},
$\mathbf{g}(\boldsymbol{\xi}, \boldsymbol{\xi}) \in \mathcal{G}(M)$.
For functions $f \in \mathcal{G}(M)$ we have the following notion of
strict positivity:
\[
f>0\ :\Longleftrightarrow \forall K \subset\subset M, \, \exists m
\in \N:\ \inf_{p \in K} f_\eps(p) \geq \eps^m, \quad \mbox{as $\eps
\rightarrow 0$}
\]
and we define time-like, null and space-like for $\boldsymbol{\xi}$
by demanding $\mathbf{g}(\boldsymbol{\xi}, \boldsymbol{\xi}) < 0$,
$\mathbf{g}(\boldsymbol{\xi}, \boldsymbol{\xi}) = 0$, respectively
$\mathbf{g}(\boldsymbol{\xi}, \boldsymbol{\xi}) > 0$.
(See~\cite{algfound} for details, as well as for a general account
of basic Lorentzian geometry in the present setting.)

\subsection{A class of metrics}
\label{definitions2}

We are now ready to define the class of metrics that we will study.
Let $(M, \mathbf{g})$ be a generalised space-time, and
$\mathbf{g}_{\eps}$ a representative of the generalised metric. Let
$p \in M$, $U$ a relatively compact open neighbourhood of $p$, and
let $t: U \rightarrow \mathbb{R}$ be a smooth map with the
properties that $t(p) = 0$, $dt \neq 0$ on $U$. We assume that there
exists an $M_0 > 0$ with $\mathbf{g}_{\eps}^{-1}(dt, dt) \le
-1/M_0^2$, as $\eps \rightarrow 0$ on $U$. Therefore the level sets
of the function $t$, $\Sigma_{\tau} := \{ q \in U : t(q) = \tau \}$,
are space-like hyper-surfaces with respect to the representative
metrics, $\mathbf{g}_{\eps}$, uniformly as $\eps \rightarrow 0$. We
define the normal covector field to these hyper-surfaces $\bsig := -
dt \in \Omega^1(U)$ which, via the constant embedding, may also be
viewed as a generalised covector field on $U$. We define the
corresponding generalised normal vector field, $\bxi$, by its
representative $\bxieps \in \mathfrak{X}(U)$, given, for each
$\eps$, by $\bsig = \mathbf{g}_{\eps}(\bxieps, \cdot)$. We now
define the generalised function, $V$, on $U$ by its representative
$V_{\eps}: U \rightarrow \mathbb{R}^+$, given by
\[
V_{\eps}^2 = - \mathbf{g}_{\eps}(\bxieps, \bxieps).
\]
We will also require the corresponding normalised versions of the
generalised normal vector field, $\bxihat = [(\bxiepshat)_\eps] =
[(\bxieps / V_{\eps})_{\eps}]$, and covector field, $\bsighat =
[(\bsigepshat)_{\eps}] = \mathbf{g}( \bxihat, \cdot)$. Observe that,
although $\bsig$ does not depend on $\eps$, the quantities derived
from it, i.e., $\bsigepshat$, $\bxieps$ and $\bxiepshat$ necessarily
do, since we are dealing with a generalised metric. Using these
quantities, one may construct a positive-definite metric associated
with the generalised space-time (cf.~\cite[Sec.\ 4]{algfound}). In
particular, we define
\[
\mathbf{e}_{\eps} := \mathbf{g}_{\eps} + 2 \bsigepshat \otimes
\bsigepshat,
\]
which clearly, for each fixed $\eps$, is a Riemannian metric on $U$.
Additionally, the resulting class $\mathbf{e} =
[(\mathbf{e}_{\eps})_\eps]$ defines a generalised Riemannian metric
on $U$ (\cite[Prop.\ 4.3]{algfound}).

We denote by $\Sigma := \Sigma_0$ the three-dimensional space-like
hypersurface through $p$. Let $\mathbf{m}$ be a background
Riemannian metric on $U$ and denote by $\|\;\|_{\mathbf{m}}$ the
norm induced on the fibers of the respective tensor bundle on $U$.
We demand the following conditions.

\begin{enumerate}
\label{settingmetric}
\item[(A)]
\label{new-setting1} For all $K$ compact in $U$, for all orders of
derivative $k \in \mathbb N_0$ and all $k$-tuples of vector fields
$\boldsymbol{\eta}_1, \dots, \boldsymbol{\eta}_k \in \mathfrak X(U)$
and for any representative $( \mathbf{g}_{\eps} )_\eps$ we have:
\begin{itemize}
\item
$\sup_K \left\Vert \mathscr{L}_{\boldsymbol{\eta}_1} \dots
\mathscr{L}_{\boldsymbol{\eta}_k} \mathbf{g}_{\eps}
\right\Vert_{\mathbf{m}} = O(\eps^{-k})\quad\quad (\eps\rightarrow
0)$;
\item
$\sup_K \left\Vert
\mathscr{L}_{\boldsymbol{\eta}_1}\dots\mathscr{L}_{\boldsymbol{\eta}_k}
\mathbf{g}_{\eps}^{-1} \right\Vert_{\mathbf{m}} =
O(\eps^{-k})\quad\quad (\eps\rightarrow 0)$.
\end{itemize}
In particular, this implies (for $k=0$) that the metrics
$\mathbf{g}_{\eps}$ and their inverses $\mathbf{g}_{\eps}^{-1}$ are
locally uniformly bounded with respect to $\eps$.

\item[(B)]
\label{new-setting2} For all $K$ compact in $U$, we have
\begin{equation}
\sup_K \left\Vert \nabla^{\eps} \boldsymbol{\xi}_{\eps}
\right\Vert_{\mathbf{m}} = O(1), \quad\quad (\eps \rightarrow 0),
\label{conditionB}
\end{equation}
where $\nabla^{\eps}$ denotes the covariant derivative with respect
to the Lorentzian metric $\mathbf{g}_{\eps}$.

\item[(C)]
\label{new-setting3} For each representative
$(\mathbf{g}_{\eps})_\eps$ of the metric $\mathbf{g}$ on $U$,
$\Sigma$ is a past-compact space-like hypersurface such that
$\partial J_\eps^+(\Sigma)=\Sigma$. Here $J_\eps^+(\Sigma)$ denotes
the topological closure of the future emission
$D^+_\eps(\Sigma)\subset U$ of $\Sigma$ with respect to
$\mathbf{g}_{\eps}$. Moreover, there exists a nonempty open set $A
\subseteq \mathcal{M}$ and an $\eps_0$ such that
\[
A \subseteq \bigcap_{\eps < \eps_0} J_\eps^+(\Sigma).
\]
(Note that we are here following the notation of
Friedlander~\cite{FL1}. The set $D^+_\eps(\Sigma)$ would be denoted
$I^+_\eps(\Sigma)$ in the conventions of, for example, Hawking and
Ellis~\cite{HE}.)
\end{enumerate}
A generalised metric with the above properties will be referred to
as a \emph{weakly singular metric}.

\vskip .2cm
\noindent{}Some comments are in order:
\begin{itemize}
\item Condition~(A) is independent of the background Riemannian metric, $\mathbf{m}$, chosen on $U$, and may be rephrased in terms of a fixed but arbitrary coordinate system, $\{ x^a \}$, on a neighbourhood of $p$ as follows. With $k=0$, Condition~(A) states that the components of the metric $\mathbf{g}_{\eps}$ and its inverse are locally uniformly bounded on $U$. In particular, the Lorentzian norm of the generalised normal vector field $\boldsymbol{\xi}_{\eps}$ may be assumed to satisfy
\begin{equation}
\label{setting1} \frac{1}{M_0} \leq
\sqrt{-g_\eps(\boldsymbol{\xi}_{\eps}, \boldsymbol{\xi}_{\eps})} =
V_\eps \leq M_0.
\end{equation}
For $k>0$, Condition~(A) states that there exists $M_k$ such that
\[
\left| \frac{\partial^k g_{ab}^{\eps}}{\partial x^{a_1} \cdots
\partial x^{a_k}} \right| \leq \frac{M_k}{\eps^k}, \quad \left|
\frac{\partial^k g_{\eps}^{ab}}{\partial x^{a_1} \cdots \partial
x^{a_k}} \right| \leq \frac{M_k}{\eps^k}.
\]

\item Conditions~(A) and~(B) imply that the generalised Riemannian metric, $\mathbf{e}$, obeys the asymptotic condition
\[
\left\Vert \nabla^\eps{} \mathbf{e}_{\eps}^{-1}
\right\Vert_{\mathbf{m}} = O(1) \qquad \mbox{ as } \qquad \eps
\rightarrow 0.
\]
{}From condition~(A), it follows that $\left\Vert \mathbf{e}_{\eps}
\right\Vert_{\mathbf{m}}, \left\Vert \mathbf{e}_{\eps}^{-1}
\right\Vert_{\mathbf{m}} = O(1)$. (This can most easily be deduced from the
form of the metric given, below, in~\eqref{coordmetric}.)
Therefore,
by the Cauchy-Schwarz inequality for the inner product induced by $\mathbf{m}$
on the bundle of $(2, 1)$ tensors on $M$, we have
\begin{equation}
\left\Vert \nabla^\eps{} \mathbf{e}_{\eps}^{-1}
\right\Vert_{\mathbf{e}_{\eps}} = O(1) \qquad \mbox{ as } \qquad
\eps \rightarrow 0.
\label{nablae}
\end{equation}
Similarly, \lq\lq lowering the index\rq\rq\ on
$\boldsymbol{\xi}_{\eps}$ in equation~\eqref{conditionB} implies
that $\sup_K \left\Vert \nabla^{\eps} \boldsymbol{\sigma}
\right\Vert_{\mathbf{m}} = O(1)$ as $\eps \rightarrow 0$, where we
have again used the fact that $\left\Vert \mathbf{g}_{\eps}
\right\Vert_{\mathbf{m}} = O(1)$. Taking the symmetric part of
$\nabla^{\eps} \boldsymbol{\sigma}$ implies that $\left\Vert
\mathscr{L}_{\boldsymbol{\xi}_{\eps}} \mathbf{g}_{\eps}
\right\Vert_{\mathbf{m}} = O(1)$ as $\eps \rightarrow 0$. Finally,
again using the Cauchy-Schwarz inequality, we deduce that
\begin{equation}
\left\Vert \mathscr{L}_{\boldsymbol{\xi}_{\eps}} \mathbf{g}_{\eps}
\right\Vert_{\mathbf{e}_{\eps}} = O(1) \qquad \mbox{ as } \qquad
\eps \rightarrow 0. \label{nablaxi}
\end{equation}
These estimates will be required in Section~\ref{partb}.

\item Condition~(C) is necessary to ensure existence of smooth solutions on the level of representatives on a
common domain (cf.\ Step 1 in the proof of Theorem~\ref{mainthm} in
Section~\ref{proof}).
\end{itemize}

\begin{remark}
\label{ConditionCRemark}
Conditions~(A) and~(B) are given in terms of the $\eps$-asymptotics of
the generalised metric. There is, however, the following close
connection to the classical situation. Assume that we are given a space-time
metric that is locally bounded but not necessarily $\Cc^{1, 1}$ or of
Geroch-Traschen class \cite{GT} (i.e., the largest class that allows
a \lq\lq reasonable\rq\rq\ distributional treatment). We may then embed this
metric into the space of generalised metrics by convolution with a
standard mollifier (cf. Section~\ref{geomth}). {}From the explicit form of the embedding it is then
clear that condition (A) holds.

We recall that in the special version of Colombeau's
construction the embedding is non-geometric and we could -- at the price
of technical complications -- resort to the full version where a
geometric embedding is available (as was done in~\cite{VW}).
Nevertheless, in the full construction generalised functions that are
embeddings of locally bounded functions still display the
$\eps$-asymptotics of condition~(A). Moreover our approach using the
special version offers more flexibility: Whenever one succeeds, e.g.\ by
using some physically motivated procedure, to model a singular metric by
a sequence of classical metrics obeying~(A)-(C), then our results apply.

Condition~(B) on the other hand demands somewhat better asymptotics of
the derivatives of the $(0, 0)$-component of the metric in adapted
coordinates (see also~\eqref{coeffestimates} below). This is a technical
condition that is satisfied by several relevant examples (see below).

As to condition~(C), the only part that exceeds the classical condition for existence and uniqueness of solutions
is the existence of the non-empty open set $A$. Geometrically, this means that the light-cones of the metric $\mathbf{g}_{\eps}$
do not collapse as $\eps \rightarrow 0$. In terms of regularisations of classical metrics, this condition will always be satisfied
if the classical metric is non-degenerate.
\end{remark}

\begin{examples}
To begin with we discuss the conical space-times of~\cite{VW}. They fall into our class since estimates~(6) and~(7) in~\cite{VW} for the embedded metric imply our condition~(A), while~(B) is immediate from the staticity of the metric.

\vskip .2cm
The metric of impulsive $pp$-waves (in \lq\lq Rosen form\rq\rq) fall into
our class. For simplicity we only consider plane waves of constant linear polarisation, i.e,
\[
 - du dv + \left( 1 + u_+ \right)^2 dx^2 + \left( 1 - u_+ \right)^2 dy^2,
\]
where $u_+ := u H(u)$ denotes the kink function. This metric is
locally bounded (actually continuous) and, since the non-trivial
behaviour involves simply the spatial part of the metric, will
therefore obey Conditions~(A) and (B) when embedded with a
standard mollifier, or -- more generally -- if we use any other
regularisation that converges at least locally uniformly to the
original metric.

\vskip .2cm
Similarly, in~\cite{PG}, metrics for expanding spherical impulsive
waves of the form
\[
2 du dv + 2 v^2 \left| dz + \frac{u_+}{2v} \overline{H}
d\overline{z} \right|^2
\]
were studied, where $H(z)$ is the Schwartzian derivative of any
arbitrary analytic function $h(z)$. Again, this metric is continuous
and the non-trivial behaviour occurs in only the spatial directions.
So we obtain conditions~(A) and~(B) as for the above case.

\vskip .2cm
In all of these examples, the discussion at the end of Remark~\ref{ConditionCRemark} imply that Condition~(C) is also satisfied.
\end{examples}

\section{The main result}
\label{mainresult}

We are interested in the initial value problem for the wave equation
\begin{eqnarray}
\Box u &=& 0 \nonumber
\\
\label{weqofsetting} u|_\Sigma &=& v
\\
\nonumber \mathscr{L}_{\widehat{\boldsymbol{\xi}}} u|_\Sigma &=& w
\end{eqnarray}
on the subset $U$ of a weakly singular space-time $(M, \mathbf{g})$
(i.e., $\mathbf{g}$ subject to the assumptions (A)--(C) of
Section~\ref{definitions2}). Here $\Sigma := \Sigma_0$ denotes the
level set $\{ q \in U: t(q) = 0 \}$ of the function $t: U
\rightarrow \mathbb{R}$ introduced in Section~\ref{definitions2}.
The initial conditions are defined by $v$ and $w$, which are given
functions in $\mathcal{G}(\Sigma)$. Note that this, in particular,
includes the case of arbitrary distributional initial data. We are
interested in finding a local solution $u \in \mathcal{G}$ on $U$
resp.\ an open subset thereof.

A general strategy to solve PDEs in $\mathcal{G}$ is the following.
First, solve the equation for fixed $\eps$ in the smooth setting and
form the net $(u_\eps)_\eps$ of smooth solutions. This will be a
candidate for a solution in $\mathcal{G}$, but particular care has
to be taken to guarantee that the $u_\eps$ share a common domain of
definition. In the second step, one shows that the solution
candidate $(u_\eps)_\eps$ is a moderate net, hence obtaining
existence of the solution $[(u_\eps)_\eps]$ in $\mathcal{G}$.
Finally, to obtain uniqueness of solutions, one has to prove that
changing representatives of the data or the metric leads to a
solution that is still in the class $[(u_\eps)_\eps]$. Note that
this amounts to an additional stability of the equation with respect
to negligible perturbations of the initial data and the metric.

According to this strategy, given the point $p$ in $\Sigma$ we may,
without loss of generality, assume that $(U, \{ x^a \})$, with
$(x^a)_{a = 0, 1, 2, 3} = (t, x^i)$ is a coordinate neighbourhood of
$p$, and formulate the initial value problem~\eqref{weqofsetting} in
terms of representatives on $U$. To this end, given a representative
$(\mathbf{g}_{\eps})_{\eps}$ of the metric $\mathbf{g}$, there exist
functions $h^{\eps}_{ij}, N_{\eps}^i$ on $U$ such that
\begin{equation}
\mathbf{g}_{\eps} = - V_{\eps}^2 dt^2 + h^{\eps}_{ij} \left( dx^i -
N_{\eps}^i dt \right) \otimes \left( dx^j - N_{\eps}^j dt \right).
\label{coordmetric}
\end{equation}
We further choose representatives $(v_\eps)_\eps$, $(w_\eps)_\eps$
of the data and a negligible net $(f_\eps)_\eps$ on $U$. We then
consider the initial value problem
\begin{eqnarray}
\Box^\eps u_\eps &=& f_\eps \nonumber
\\
\label{weqofsettingeps} u_\eps(t=0, x^i) &=& v_\eps(x^i)
\\
\nonumber \mathscr{L}_{\widehat{\boldsymbol{\xi}}_{\eps}} u_\eps
(t=0, x^i) &=& w_\eps(x^i),
\end{eqnarray}
where $\Box^\eps$ is the d'Alembertian derived from our particular
representative $\mathbf{g}_{\eps}$, i.e.,
\begin{align*}
\Box^\eps u_\eps &= \vert g_{\eps} \vert^{-1/2} \partial_a \left(
\vert g_{\eps} \vert^{1/2} g_{\eps}^{ab} \partial_b u_\eps \right)
\\
&= -\frac{1}{V_\eps^2} \partial_t^2 u_\eps - \frac{2}{V_{\eps}^2}
N^i \partial_t \partial_i u_{\eps} + \left( h_{\eps}^{ij} -
\frac{1}{V_{\eps}^2} N_{\eps}^i N_{\eps}^j \right) \partial_i
\partial_j u_{\eps} - g_{\eps}^{ab} \,
\Gamma[\mathbf{g}_{\eps}]^c{}_{ab} \, \frac{\partial
u_{\eps}}{\partial x^c},
\end{align*}
and $h_{\eps}^{ij}$ are the components of the inverse of
$h^{\eps}_{ij}$, $g_{\eps} := \det \mathbf{g}_{\eps}$, and
$\Gamma[\mathbf{g}_{\eps}]^c{}_{ab}$ denote the Christoffel symbols
of the metric $\mathbf{g}_{\eps}$. Note that, by Conditions~(A)
and~(B) of Section~\ref{definitions2}, the following asymptotic
estimates hold for the components of the metric in the above
coordinate system
\begin{equation}
\left.
\begin{aligned}
&V_{\eps}, h^{\eps}_{ij}, N_{\eps}^i = O(1)
\\
&\partial_a V_{\eps} = O(1)
\\
&\partial^{\alpha} V_{\eps}, \partial^{\alpha} h^{\eps}_{ij},
\partial^{\alpha} N_{\eps}^i = O({\eps}^{-|\alpha|})  \mbox{ for all
multi-indices $\alpha$ with $|\alpha| \ge 1$}
\end{aligned}
\right\} \qquad \mbox{ as $\eps \rightarrow 0$}.
\label{coeffestimates}
\end{equation}
Following the general strategy outlined above, we will prove local
unique solvability of~\eqref{weqofsetting} by showing that the
smooth solutions, $(u_\eps)_\eps$, of~\eqref{weqofsettingeps} form a
moderate net, and hence determine a class in $\mathcal{G}$, and that
this class is independent of the choice of representatives of $v$,
$w$ and $\mathbf{g}$. More precisely, our main result is the
following:

\begin{theorem}[Local existence and uniqueness of generalised solutions]
\label{mainthm} 
Let $(M, \mathbf{g})$ be a generalised space-time
and assume that Conditions~(A)--(C) of Section~\ref{definitions2}
hold. Then, for each $p \in \Sigma$, there exists an open
neighbourhood $V$ on which the initial value problem for the wave
equation~\eqref{weqofsetting} has a unique solution in
$\mathcal{G}(V)$.
\end{theorem}

We split the proof in a series of arguments, the core of which are
higher order energy estimates. To prepare for these, we first
introduce suitable energy tensors and energy integrals.

\section{Energy integrals}
\label{eint}

By assumption, we have a point $p \in M$ and an open neighbourhood
of $p$, $U$, and a map $t: U \rightarrow \mathbb{R}$ with $t(p) = 0$
such that $U$ is foliated by the level sets of the function $t$,
$\Sigma_{\tau} := \{ q \in U : t(q) = \tau \}, \tau \in [-\gamma,
\gamma]$, for some $\gamma > 0$. Moreover, the level sets
$\Sigma_\tau$ are space-like with respect to the generalised metric
$\mathbf{g}$. We now consider solving the forward in time initial
value problem for the wave equation on $U$, i.e., with $\tau \geq 0$
(see Figure~\ref{fig:1}).
\begin{figure}[ht!]
\begin{center}
\fbox{\includegraphics[width=5.21in]{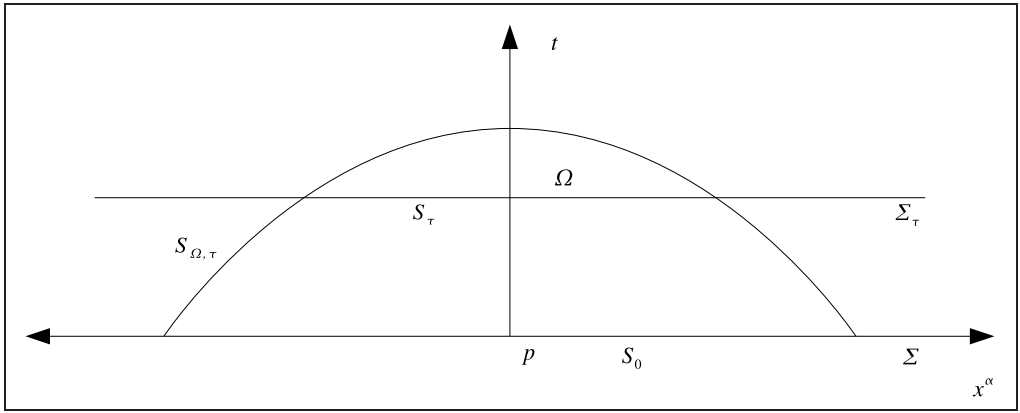}}
\end{center}
\caption{Local foliation of space-time} \label{fig:1}
\end{figure}

Given $p \in \Sigma = \Sigma_0$, let $\Omega$ be a neighbourhood of
$p$ with the properties that $\overline{\Omega} \subset U$, and such
that the boundary of the region $\Omega \cap \{ q \in U: t(q) \ge 0
\}$ is space-like\footnote{The existence of such a set, $\Omega$,
follows from the fact that $\Vert \mathbf{g}_{\eps}^{-1}
\Vert_{\mathbf{m}} = O(1)$ as $\eps \rightarrow 0$. Geometrically,
this condition means that the collection of timelike directions at a
given point is not collapsing to the empty set.}. We denote by
$S_\tau := \Sigma_\tau \cap \Omega$ and by $\Omega_\tau$ the open
part of $\Omega$ between $\Sigma$ and $\Sigma_\tau$. We denote the
part of the boundary of $\Omega_{\tau}$ with $0 \le t \le \tau$ by
$S_{\Omega, \tau}$, so that $\partial \Omega_{\tau} = S_0 \cup
S_{\tau} \cup S_{\Omega, \tau}$.

\vskip .2cm
\noindent{\textbf{Notation}}. In order to simplify calculations,
from now on we will adopt abstract index notation for (generalised)
tensorial objects (see, e.g.,~\cite{PR1}). In particular,
representatives of the metric $\mathbf{g}_{\eps}$ and its inverse
will be denoted by $g^{\eps}_{ab}$ and $g_{\eps}^{ab}$,
respectively, and similarly for the corresponding Riemannian metric
$\mathbf{e}_{\eps}$. We denote the representative of the generalised
normal vector field, $\boldsymbol{\xi}$, by $\xi_{\eps}^a$, and the
corresponding generalised covector field, $\boldsymbol{\sigma}$, by
$\xi^{\eps}_a$. In addition, to simplify the notation for tensors we
are going to use capital letters to abbreviate tuples of indices,
i.e., we will write $T^I_J$ for $T^{p_1\dots p_r}_{q_1\dots q_s}$
with $|I|=r$, $|J|=s$. Also for $I$, $J$ of equal length, say $r$,
we write $e_{IJ}$ for $e_{p_1q_1}e_{p_2q_2} \dots e_{q_rp_r}$.

\vskip .2cm
We now use the Riemannian metric $\mathbf{e}_{\eps}$ and the
covariant derivative with respect to $\mathbf{g}_{\eps}$ --- which
we have denoted by $\nabla^\eps$ --- to define $\eps$-dependent
Sobolev norms on $U$.

\begin{definition}[Sobolev norms]
Let $T^I_J$ be a smooth tensor field and $u$ a smooth function on
$U$, $\eps > 0$, $0 \leq \tau \leq \gamma$ and $k, j \in \N_0$.
\begin{enumerate}
\item We define the \lq\lq pointwise\rq\rq\ norm of $T^I_J$ by
\[
\|T^I_J\|_{e_\eps}^2 := e^\eps_{KL} e_\eps^{IJ} T^K_I T^L_J
\]
and the \lq\lq pointwise norm\rq\rq\ of covariant derivatives of $u$
by
\[
|\nabla_\eps^{(j)}u|^2
:=||\nabla^\eps_{p_1}\dots\nabla^\eps_{p_j}u||^2_{e_\eps}=e^{p_1q_1}_\eps
\dots e^{p_jq_j}_\eps \left( \nabla^\eps_{p_1} \dots
\nabla^\eps_{p_j} u \right) \left( \nabla^\eps_{q_1} \dots
\nabla^\eps_{q_j} u \right).
\]
\item On $\Omega_\tau$ we define Sobolev norms with respect to $\nabla_a^\eps$ resp.\ partial derivatives by
\begin{eqnarray*}
\SobODt{u}{k}&:=&
\left(\sum_{j=0}^k \int_{\Omega_\tau}|\nabla_\eps^{(j)}(u)|^2\mu^\eps\right)^{\frac{1}{2}}\\
\SobOdt{u}{k}&:=& \left( \sum_{\stack{p_1, \dots, p_j}{0\leq j\leq
k}}
\int_{\Omega_\tau}|\partial_{p_1}\dots\partial_{p_j}u|^2\mu^\eps\right)^{\frac{1}{2}},
\end{eqnarray*}
where $\mu^\eps$ denotes the volume form derived from
$\mathbf{g}_{\eps}$.
\item The respective \lq\lq three-dimensional\rq\rq\ Sobolev norms are defined by
\begin{eqnarray*}
\SobSDt{u}{k}&:=&
\left(\sum_{j=0}^k \int_{S_\tau}|\nabla_\eps^{(j)}(u)|^2\mu_\tau^\eps\right)^{\frac{1}{2}}\\
\SobSdt{u}{k}&:=& \left(\sum_{\stack{p_1, \dots, p_j}{0 \leq j \leq
k}} \int_{S_\tau}|\partial_{p_1}\dots\partial_{p_j}
u|^2\mu_\tau^\eps\right)^{\frac{1}{2}},
\end{eqnarray*}
where $\mu_\tau^\eps$ is the unique three-form induced on $S_\tau$
by $\mu^\eps$ such that $dt \wedge \mu_\tau^\eps = \mu^\eps$ holds
on $S_\tau$. Note that although the integration is performed over
the three-dimensional manifold $S_\tau$ only, derivatives are not
confined to directions tangential to $S_\tau$.
\end{enumerate}
\end{definition}

Observe that, due to the use of a generalised metric, even the norms
$\SobSdt{u}{k}$ depend on $\eps$. However, due to Condition~(A),
with $k=0$, they are equivalently to an $\eps$-independent norm
derived, for example, from the fixed background metric $\mathbf{m}$.
In the following, we will provide suitable higher order energy
estimates for nets of solutions of the wave equation. These
estimates are best expressed in terms of energy momentum tensors and
energy integrals, which we define following~\cite[Sec.\ 4]{VW}. For
the \lq\lq classical\rq\rq\ case, see~\cite[Sec.\ 7.4]{HE},
\cite[Sec.\ 4.4]{ClarkeBook} and~\cite{clemens} for a recent review.

\begin{definition}[Energy momentum tensors and energy integrals]
Let $u \in \Cc^\infty(U)$ and $k \in \N_0$. On $\Omega$ we define
\begin{enumerate}
\item the energy momentum tensors by ($k > 0$)
\begin{eqnarray*}
T^{ab, 0}_{\eps}(u) &:=&-\frac{1}{2}g^{ab}_\eps u^2,
\\
T^{ab, k}_{\eps}(u) &:=&\big( g^{ac}_\eps
g^{bd}_\eps-\frac{1}{2}g^{ab}_\eps g^{cd}_\eps \big)e^{p_1q_1}_\eps
\dots e^{p_{k-1}q_{k-1}}_\eps
(\nabla_c^\eps \nabla_{p_1}^\eps \dots \nabla_{p_{k-1}}^\eps u)
(\nabla_d^\eps \nabla_{q_1}^\eps \dots \nabla_{q_{k-1}}^\eps u),
\end{eqnarray*}
\item the energy integrals by
\begin{equation}
E^k_{\tau, \eps}(u):= \sum_{j=0}^k \int_{S_\tau} T^{ab, j}_\eps(u)
\xi_a \widehat{\xi}_b^\eps \mu_\tau^\eps, \qquad k \ge 0.
\label{HEintegrals}
\end{equation}
\end{enumerate}
\end{definition}

It may be verified, by direct calculation, that the tensor fields
$T^{ab, k}_{\eps}(u)$ satisfy the dominant energy condition. Indeed
it suffices to observe that, for any future-directed time-like
vector field $\mathbf{U}$, the expression $U^a U^b - \frac{1}{2}
\mathbf{g}_{\eps}(\mathbf{U}, \mathbf{U}) g_{\eps}^{ab}$ defines a
Riemannian metric for fixed $\eps$. For details, see Proposition~3.6
of~\cite{PhDthesis}. For a generalised formulation of the dominant
energy condition, see~\cite{algfound}.

\begin{remark}
The energy momentum tensors introduced above are related to the
super-energy tensors of Senovilla~\cite{Sen}. Omitting indices and $\eps$'s for the moment, we construct the super-energy tensor, $S^k$, of type
$(0, 2k)$ (see Definition 3.1 in~\cite{Sen}). Then $T^k(u)$ are
the $(2k-2)$-fold contraction of $S^k$ with the time-like vector
field $\boldsymbol{\xi}$. Theorem 4.1 of~\cite{Sen} then implies
that the tensors $S^k (k \ge 0)$ satisfy the dominant super-energy
property, from which it follows more elegantly that the $T^k(u)$ satisfy the dominant
energy condition.
\end{remark}

\begin{remark}
The energy integrals may be written in the more symmetrical form
\[
E^k_{\tau, \eps}(u):= \sum_{j=0}^k \int_{S_\tau} T^{ab, j}_\eps(u)
\widehat{\xi}_a^{\eps} \widehat{\xi}_b^\eps \widehat{\mu}_\tau^\eps,
\]
where we have defined the volume element $\widehat{\mu}_\tau^\eps =
V_{\eps}^{-1} \mu_\tau^\eps$ on $S_{\tau}$. In terms of the
decomposition of the metric given in equation~\eqref{coordmetric},
$\widehat{\mu}_\tau^\eps$ is the volume form on $S_{\tau}$ defined
by the three-dimensional metric $\mathbf{h}_{\eps} := h^{\eps}_{ij}
dx^i \otimes dx^j$.
\end{remark}

Since the part $S_{\Omega,\tau}$ of the  boundary of $\Omega$ is
space-like and $T^{ab, j}_\eps(u)$ satisfies the dominant energy
condition, an application of the Stokes theorem yields
\begin{eqnarray*}
\int_{\Omega_\tau} \nabla_a^\eps  \left( T^{ab, j}_\eps(u) \xi_b
\right)\mu_\eps
&=& \int_{S_\tau} T^{ab, j}_\eps(u) \xi_b \widehat{\xi}_a^{\eps}
\mu_\tau^\eps
 -\int_{S_0} T^{ab, j}_\eps(u) \xi_b \widehat{\xi}_a^{\eps} \mu_0^\eps
 +\int_{S_{\Omega, \tau}} T^{ab, j}_\eps(u) \xi_b n_a^{\eps} dS_\eps\\
&\ge&  \int_{S_\tau} T^{ab, j}_\eps(u) \xi_b \widehat{\xi}_a^{\eps}
\mu_\tau^\eps
 -\int_{S_0} T^{ab, j}_\eps(u) \xi_b \widehat{\xi}_a^{\eps} \mu_0^\eps,
\end{eqnarray*}
where $\mathbf{n}^{\eps}$  and $dS_{\eps}$ denote the unit normal
and surface element on $\partial \Omega_\tau$, respectively. Hence
summing over $j$ we have the following energy inequality for each
$\eps > 0$ and each $0 \leq \tau \leq \gamma$
\begin{equation}
\label{energyhierarchystokes} E^k_{\tau, \eps}(u) \leq E^k_{\tau=0,
\eps}(u) + \sum_{j=0}^k \int_{\Omega_\tau} \left( \xi_b
\nabla_a^\eps T^{ab, j}_\eps(u) + T^{ab, j}_\eps(u) \nabla_a^\eps
\xi_b \right)\mu_\eps.
\end{equation}

Note that the energy integrals and foliation used here correspond
closely to those used in~\cite[Sec.\ 4.3]{HE}.
In~\cite[pp.~1341]{VW}, due to a different choice of foliation,
inequality~\eqref{energyhierarchystokes} is replaced with an
equality. This alternative foliation allows one to work without the
explicit use of the dominant energy condition, but complicates some
of the resulting energy estimates.

To end this section, we prove the equivalence of the Sobolev norms
and the energy integrals. Note that this result is the analogue of
Lemma~1 in~\cite{VW} for our class of metrics, and is one of the key
estimates in our approach.

\begin{lemma}[Energy integrals and Sobolev norms]
\label{lemma1}
\quad\\
\begin{enumerate}
\item
There exist constants $A, A'$ such that for each $k\geq 0$
\begin{equation}
\label{ineqEXSD} A'(\SobSDt{u}{k})^2 \leq E^k_{\tau, \eps}(u) \leq
A(\SobSDt{u}{k})^2
\end{equation}
\item For each $k\geq 1$, there exist positive constants $B_k, B_k'$ such that
\begin{eqnarray}
\label{ineqSDXSd}
(\SobSDt{u}{k})^2&\leq&B_k'\sum_{j=1}^k\frac{1}{\eps^{2(k-j)}}(\SobSdt{u}{j})^2
\\
\label{ineqSdXSD}
(\SobSdt{u}{k})^2&\leq&B_k\sum_{j=1}^k\frac{1}{\eps^{2(k-j)}}(\SobSDt{u}{j})^2
\end{eqnarray}
For $k=0$ we simply have $(\SobSDt{u}{0})^2=(\SobSdt{u}{0})^2$.
\end{enumerate}
\end{lemma}

\begin{proof}
(1): For $k=0$ we have
\[
T^{ab, 0}_\eps(u) \xi_a \widehat{\xi}_b^{\eps}= -\frac{1}{2}
g^{ab}_\eps \xi_a \widehat{\xi}_b^{\eps}u^2
=-\frac{1}{2}\sqrt{-\mathbf{g}_\eps(\boldsymbol{\xi}_\eps,
\boldsymbol{\xi}_\eps)}u^2=\frac{V_\eps}{2}u^2,
\]
hence by~\eqref{setting1}, setting $A:=M_0/2$ and $A':=1/({2 M_0})$,
we obtain
\[
A'u ^2 \leq T^{ab, 0}_\eps(u) \xi_a \widehat{\xi}_b^{\eps} \leq A
u^2,
\]
which upon integrating gives the result.

For the case $k>0$ note that
\begin{eqnarray*}
(g^{ac}_\eps g^{bd}_\eps-\frac{1}{2}g^{ab}_\eps g^{cd}_\eps) \xi_a
\widehat{\xi}_b^\eps &=& \frac{1}{2} V_\eps \left(g^{cd}_\eps +
\frac{2}{V_\eps^2} \xi_\eps^c \xi_\eps^d \right) \, =\,
\frac{1}{2}V_\eps e^{cd}_\eps.
\end{eqnarray*}
Hence, we may write
\begin{align*}
T^{ab, j}_{\eps}(u) \xi_a \widehat{\xi}_b^\eps &= \frac{1}{2} V_\eps
e^{cd}_\eps e^{p_1q_1}_\eps \dots e^{p_{j-1}q_{j-1}}_\eps
(\nabla_c^\eps \nabla_{p_1}^\eps \dots \nabla_{p_{j-1}}^\eps u)
(\nabla_d^\eps \nabla_{q_1}^\eps \dots\nabla_{q_{j-1}}^\eps u)
\\
&= \frac{1}{2} V_\eps | \nabla^{(j)}_\eps u|^2.
\end{align*}
Using~(A), this implies that
\[
A'| \nabla^{(j)}_\eps u|^2 \leq T^{ab, j}_\eps(u) \xi_a
\widehat{\xi}_b^\eps \leq A| \nabla^{(j)}_\eps u|^2,
\]
which upon summation and integration establishes the claim.

\vskip .2cm
(2) follows by (A) from the fact that on the compact closure of
$\Omega$ the metrics $\mathbf{e}_{\eps}$ and $\delta_{ab}$ are
equivalent and the Christoffel symbols and its derivatives are
bounded by the respective inverse powers of $\eps$.
\end{proof}

\section{Energy estimates}
\label{partb}

In this section, we establish the core estimates needed in the proof
of our main theorem.

\begin{proposition}
\label{energyinequality} Let $u_{\eps}$ be a solution
of~\eqref{weqofsettingeps} on $U$. Then, for each $k \ge 1$, there
exist positive constants $C_k', C_k'', C_k'''$ such that for each
$0\leq\tau\leq\gamma$ we have
\begin{eqnarray}
\label{energyinequalitylevelkformula} E^k_{\tau, \eps}(u_\eps)
&\leq& E^k_{0, \eps}(u_\eps) +C_k'(\SobODt{f_\eps}{k-1})^2
+C_k''\sum_{j=1}^{k-1}\frac{1}{\eps^{2(1+k-j)}}
\int_{\zeta=0}^\tau E_{\zeta, \eps}^j(u_\eps)d\zeta\nonumber\\
&&+C_k''' \int_{\zeta=0}^\tau E_{\zeta, \eps}^k(u_\eps) d\zeta.
\end{eqnarray}
\end{proposition}

Before proving this statement, we draw the essential conclusions
from it. Observe that the constant in front of the highest order
term on the r.h.s.\ does not depend on $\eps$, hence we obtain, by
an application of Gronwall's lemma.

\begin{corollary}
\label{applicationenergygronwall} \label{energyestimate} Let
$u_{\eps}$ be a solution of~\eqref{weqofsettingeps} on $U$. Then,
for each $k \geq 1$, there exist positive constants $C_k', C_k'',
C_k'''$ such that for each $0 \leq \tau \leq \gamma$,
\begin{equation}
\label{applicationenergygronwallformula} E^k_{\tau,
\eps}(u_\eps)\leq \left(E^k_{0, \eps}(u_\eps)+C_k'
(\SobODt{f_\eps}{k-1})^2+C_k''\sum_{j=1}^{k-1}\frac{1}{\eps^{2(1+k-j)}}
\int\limits_{\zeta=0}^\tau E_{\zeta, \eps}^j(u_\eps) d\zeta\right)
e^{C_k'''\tau}
\end{equation}
\end{corollary}

This statement immediately implies the main result in this section.
\begin{corollary}
\label{energiesviainitialenergies} Let $u_{\eps}$ be a solution
of~\eqref{weqofsettingeps} on $U$. If, all $k \ge 1$, the initial
energy $(E^k_{0, \, \eps}(u_\eps))_\eps$ is a moderate resp.\
negligible net of real numbers, and $(f_\eps)_\eps$ is negligible
then
\[
\sup_{0 \leq \tau \leq \gamma} (E^k_{\tau, \, \eps}(u_\eps))_\eps
\]
is moderate resp.\ negligible.
\end{corollary}

\begin{proof}[Proof of~\ref{energyinequality}]
We begin by estimating the second integrand on the r.h.s.\ of
equation~\eqref{energyhierarchystokes}. Using the fact that the
energy tensors are symmetric then, by an application of the
Cauchy-Schwarz inequality to the inner product induced on the tensor
bundle $\mathcal{T}^{2}_0(M)$ by the metric $\mathbf{e}_{\eps}$, we
deduce that
\begin{equation}\label{modi1}
\left| T^{ab, j}_\eps(u_{\eps}) \nabla_a^\eps \xi_b \right| \le
\Vert T^{ab, j}_\eps(u_{\eps}) \Vert_{\mathbf{e}_{\eps}} \Vert
\nabla_{(a}^\eps \xi_{b)} \Vert_{\mathbf{e}_{\eps}} = \frac{1}{2}
\Vert T^{ab, j}_\eps(u_{\eps}) \Vert_{\mathbf{e}_{\eps}} \Vert
\mathscr{L}_{\boldsymbol{\xi}_{\eps}} \mathbf{g}_{\eps}
\Vert_{\mathbf{e}_{\eps}}.
\end{equation}
Equation~\eqref{nablaxi} implies that there exists a constant $K >
0$ such that $\Vert \mathscr{L}_{\boldsymbol{\xi}_{\eps}}
\mathbf{g}_{\eps} \Vert_{\mathbf{e}_{\eps}} \le K$. In the case
$j=0$, we have
\[
\Vert T^{ab, 0}_\eps(u_{\eps}) \Vert_{\mathbf{e}_{\eps}}^2 = \left(
- \frac{1}{2} \right)^2 u_{\eps}^4 e^{\eps}_{ab} e^{\eps}_{cd}
g_{\eps}^{ac} g_{\eps}^{bd} = u_{\eps}^4,
\]
so
\[
\Vert T^{ab, 0}_\eps(u_{\eps}) \Vert_{\mathbf{e}_{\eps}} =
u_{\eps}^2 = | \nabla^{(0)}_\eps u_{\eps}|^2.
\]
For $j \ge 1$, we have
\begin{align*}
\Vert T^{ab, j}_\eps(u_{\eps}) \Vert_{\mathbf{e}_{\eps}}^2 &=
\eep{aa'} \eep{bb'} \left( g^{ac}_\eps g_{\eps}^{bd} - \frac{1}{2}
g_{\eps}^{ab} g^{cd}_\eps \right) \left( g^{a'c'}_\eps
g_{\eps}^{b'd'} - \frac{1}{2} g_{\eps}^{a'b'} g^{c'd'}_\eps \right)
\\
&\hskip 3cm \times \left( \vphantom{|^|} \nabla_c \nabla_I u_{\eps}
\right) \left( \vphantom{|^|} \nabla_d \nabla_J u_{\eps} \right)
\left( \vphantom{|^|} \nabla_{c'} \nabla_{I'} u_{\eps} \right)
\left( \vphantom{|^|} \nabla_{d'} \nabla_{J'} u_{\eps} \right)
e^{IJ}_\eps e^{I'J'}_\eps
\\
&= e_{\eps}^{cc'} e_{\eps}^{dd'} \left( \vphantom{|^|} \nabla_c
\nabla_I u_{\eps} \right) \left( \vphantom{|^|} \nabla_d \nabla_J
u_{\eps} \right) \left( \vphantom{|^|} \nabla_{c'} \nabla_{I'}
u_{\eps} \right)  \left( \vphantom{|^|} \nabla_{d'} \nabla_{J'}
u_{\eps} \right) e^{IJ}_\eps e^{I'J'}_\eps
\\
&\le A_j | \nabla^{(j)}_\eps u_{\eps}|^2,
\end{align*}
where $A_j \sim 4^j$ are combinatorial constants. Letting $A_0 :=
1$, we deduce that
\[
\left| T^{ab, j}_\eps(u_{\eps}) \nabla_a^\eps \xi_b \right| \le
\frac{1}{2} A_j K | \nabla^{(j)}_\eps u_{\eps}|^2, \quad \mbox{for
$j \ge 0$.}
\]
Letting $\tilde{A_k} := \max_{j = 0, \dots, k} A_k$, we therefore
find that
\begin{equation}
\left| \sum_{j=0}^k \int_{\Omega_\tau} T^{ab, j}_\eps(u_{\eps})
\nabla_a^\eps \xi_b \mu_\eps \right| \le \frac{1}{2} \tilde{A}_k K
(\SobODt{u_{\eps}}{k})^2 \label{term2}
\end{equation}
for $k \ge 0$.

\vskip .2cm
We now consider the first integrand on r.h.s.\
of~\eqref{energyhierarchystokes}. Beginning with the case $k=1$, the
divergence terms that we require take the form
\begin{eqnarray*}
\nabla_a^\eps T^{ab, 0}_\eps(u_\eps) &=& -\frac{1}{2}(\nabla_a^\eps
g^{ab}_\eps) u_\eps^2 - \left( \frac{1}{2} g_\eps^{ab}
\right)(2u_\eps\nabla_a^\eps u_\eps)
\, =\, -u_\eps \nabla_\eps^b u_\eps\\
\nabla_a^\eps T^{ab, 1}_\eps(u_\eps) &=& \left( g^{ac}_\eps
g^{bd}_\eps - \frac{1}{2} g^{ab}_\eps g^{cd}_\eps \right)
(\nabla_a^\eps\nabla_c^\eps u_\eps\nabla_d^\eps u_\eps
+\nabla_c^\eps u_\eps\nabla_a^\eps\nabla_d^\eps u_\eps)\nonumber\\
&=&\nabla^c_\eps\nabla_c^\eps u_\eps\nabla^b_\eps u_\eps \, =\,
(\Box^\eps u_\eps) \nabla^b_\eps u_\eps \, =\, f_\eps \nabla^b_\eps
u_\eps.\label{Tab1wave}
\end{eqnarray*}
Inserting this and the $k=1$ form of~\eqref{term2}
into~\eqref{energyhierarchystokes} yields
\begin{eqnarray*}
E^1_{\tau, \eps}(u_\eps) &\leq& E^1_{0, \eps}(u_\eps) +
\int_{\Omega_\tau} \xi_{\eps}^a\nabla_a^\eps u_\eps
(f_\eps-u_\eps)\mu_\eps + \frac{1}{2} \tilde{A}_1 K
(\SobODt{u_{\eps}}{1})^2
\\
&\leq& E^1_{0, \eps}(u_\eps) +\left(
\int_{\Omega_\tau}(\xi_{\eps}^a\nabla_a^\eps
u_\eps)^2\mu_\eps\right)^{\frac{1}{2}} \left(
\int_{\Omega_\tau}|f_\eps-u_\eps|^2\mu_\eps\right)^{\frac{1}{2}}  +
\frac{1}{2} \tilde{A}_1 K (\SobODt{u_{\eps}}{1})^2
\\
&\leq& E^1_{0, \eps}(u_\eps) + M_0 \left( \int_{\Omega_\tau} |
\nabla_\eps^{(1)}u_\eps |^2\mu_\eps\right)^{\frac{1}{2}} \,
\left(\left(
\int_{\Omega_\tau}|f_\eps|^2\mu_\eps\right)^{\frac{1}{2}} +\left(
\int_{\Omega_\tau}|u_\eps|^2\mu_\eps\right)^{\frac{1}{2}}\right)
\\
& &+ \frac{1}{2} \tilde{A}_1 K (\SobODt{u_{\eps}}{1})^2
\\
&\leq& E^1_{0, \eps}(u_\eps) +\frac{M_0}{2}\left(
\int_{\Omega_\tau}\left(|\nabla_\eps^{(1)}u_\eps|^2+|u_\eps|^2\right)\mu_\eps
+ \int_{\Omega_\tau}|\nabla_\eps^{(1)}u_\eps|^2\mu_\eps+
\int_{\Omega_\tau}|f_\eps|^2\mu_\eps \right)
\\
& &+ \frac{1}{2} \tilde{A}_1 K (\SobODt{u_{\eps}}{1})^2
\\
&\leq& E^1_{0, \eps}(u_\eps) +\frac{M_0}{2}\,
\left(\SobODt{f_\eps}{0}\right)^2 + \left( M_0 + \frac{1}{2}
\tilde{A}_1 K \right)\, \left(\SobODt{u_\eps}{1}\right)^2,
\end{eqnarray*}
where we have repeatedly used the Cauchy-Schwarz inequality. Now we
use~\eqref{ineqEXSD} to obtain
\[
\left(\SobODt{u_\eps}{1}\right)^2 =
\int_{\zeta=0}^\tau(\SobSDz{u_\eps}{1})^2 d\zeta \leq\frac{1}{A'}
\int_{\zeta=0}^\tau E^1_{\zeta, \eps}(u_\eps)d\zeta.
\]
Setting $C_1' := M_0/2$, $C_1''=0$, and $C_1''' := ( M_0 +
\tfrac{1}{2} \tilde{A}_1 K )/A' = 2 M_0 ( M_0 + \tfrac{1}{2}
\tilde{A}_1 K )$ yields the claim for $k=1$.

\bigskip
We now turn to the case $k>1$. We first derive an estimate for
\[
\xi_b \nabla^\eps_a T^{ab, k}_\eps(u_\eps) = I_1 + I_2 + I_3,
\]
where we have defined
\begin{align*}
I_1 &:= \left( g^{ac}_\eps \xi_\eps^d - \frac{1}{2} \xi_\eps^a
g^{cd}_\eps \right) \left( \nabla_a^\eps e^{IJ}_\eps \right)
(\nabla_c^\eps \nabla_I^\eps u_\eps)( \nabla_d^\eps \nabla_J^\eps
u_\eps)
\\
I_2 &:= - 2 e^{IJ}_\eps \left( \nabla_d^\eps \nabla_J^\eps u_\eps
\right) \left( \xi_\eps^a g^{cd}_\eps \nabla^\eps_{[a}
\nabla^\eps_{c]} \nabla_I^\eps u_\eps \right)
\\
I_3 &:= e^{IJ}_\eps \left( \xi_\eps^d \nabla_d^\eps \nabla_J^\eps
u_\eps \right) \left( g^{ac}_\eps \nabla_a^\eps \nabla_c^\eps
\nabla_I^\eps u_\eps \right)
\end{align*}
The strategy is, again, to remove the terms involving derivatives of
order $k+1$ using the wave equation. This requires interchanging the
order of covariant derivatives, and therefore introduces additional
curvature terms. We now calculate the moduli of the terms $I_1$,
$I_2$, $I_3$ separately.

We begin by estimating $|I_1|$:
\begin{align*}
|I_1| &= \left| \left( g^{ac}_\eps \xi_{\eps}^d - \frac{1}{2}
\xi_{\eps}^a g^{cd}_\eps \right) \left( \nabla_a^\eps e^{IJ}_\eps
\right) (\nabla_c^\eps \nabla_I^\eps u_\eps)( \nabla_d^\eps
\nabla_J^\eps u_\eps) \right|
\\
&\le \left\Vert \left( g^{ac}_\eps \xi_{\eps}^d - \frac{1}{2}
\xi_{\eps}^a g^{cd}_\eps \right) \nabla_a^\eps e^{IJ}_\eps
\right\Vert_{e_{\eps}} \cdot \left\Vert (\nabla_c^\eps \nabla_I^\eps
u_\eps)( \nabla_d^\eps \nabla_J^\eps u_\eps) \right\Vert_{e_{\eps}}
\\
&= \left\Vert \left( g^{ac}_\eps \xi_{\eps}^d - \frac{1}{2}
\xi_{\eps}^a g^{cd}_\eps \right) \nabla_a^\eps e^{IJ}_\eps
\right\Vert_{e_{\eps}} \cdot \vert \nabla_{\eps}^{(k)} u_\eps
\vert^2,
\end{align*}
where the inequality in the second line results from applying the
Cauchy-Schwarz inequality to the inner product induced on the tensor
bundle $\mathcal{T}^{2k}_0(M)$ by the metric $\mathbf{e}_{\eps}$.
The square of the first term may then be evaluated as
\begin{align*}
\left\Vert \left( g^{ac}_\eps \xi_{\eps}^d - \frac{1}{2}
\xi_{\eps}^a g^{cd}_\eps \right) \nabla_a^\eps e^{IJ}_\eps
\right\Vert_{e_{\eps}}^2 &= \eep{cc'} \eep{dd'} \eep{II'} \eep{JJ'}
\left( g^{ac}_\eps \xi_{\eps}^d - \frac{1}{2} \xi_{\eps}^a
g^{cd}_\eps \right) \left( g^{a'c'}_\eps \xi_{\eps}^{d'} -
\frac{1}{2} \xi_{\eps}^{a'} g^{c'd'}_\eps \right)
\\
&\hskip 3cm \times \left( \vphantom{|^|} \nabla_a^\eps e^{IJ}_\eps
\right) \left( \nabla_{a'}^\eps e^{I'J'}_\eps \right)
\\
&= \left\Vert \xi_{\eps} \right\Vert_{e_{\eps}}^2 e_{\eps}^{aa'}
\eep{II'} \eep{JJ'} \left( \vphantom{|^|} \nabla_a^\eps e^{IJ}_\eps
\right) \left( \nabla_{a'}^\eps e^{I'J'}_\eps \right) = \left\Vert
\xi_{\eps} \right\Vert_{e_{\eps}}^2 \cdot \left\Vert \nabla_a^\eps
e^{IJ}_\eps \right\Vert_{e_{\eps}}^2.
\end{align*}
We now note that, by Condition~(A) and equation~\eqref{nablae} of
Section~\ref{definitions2}, we have,
\[
\left\Vert \nabla_a^\eps e^{IJ}_\eps \right\Vert_{e_{\eps}} =O(1),
\quad (\eps\rightarrow 0).
\]
In particular, on each compact set there exists a positive constant,
$C_k$, such that $\left\Vert \nabla_a^\eps e^{IJ}_\eps
\right\Vert_{e_{\eps}} \le C_k$, as $\eps \rightarrow 0$. Therefore,
we have the following estimate for $I_1$:
\begin{equation}
|I_1| \le C_k \cdot \left\Vert \xi_{\eps} \right\Vert_{e_{\eps}}
\cdot \vert \nabla_{\eps}^{(k)} u_\eps \vert^2 \le C_k M_0 \cdot
\vert \nabla_{\eps}^{(k)} u_\eps \vert^2, \label{ro-est1}
\end{equation}
locally, as $\eps \rightarrow 0$.

\vskip .2cm
Next we turn to $I_2$. We then have
\begin{align*}
|I_2| = \left| 2 e^{IJ}_\eps \left( \nabla_d^\eps \nabla_J^\eps
u_\eps \right) \left( \xi_{\eps}^a g^{cd}_\eps \nabla^\eps_{[a}
\nabla^\eps_{c]} \nabla_I^\eps u_\eps \right) \right| &= \left|
e^{IJ}_\eps \left( \nabla_d^\eps \nabla_J^\eps u_\eps \right) \left(
\xi_{\eps}^a e^{cd}_\eps \left[ \nabla^\eps_a, \nabla^\eps_c \right]
\nabla_I^\eps u_\eps \right) \right|
\\
&= \left| \left( e^{IJ}_\eps \xi_{\eps}^a e^{cd}_\eps \nabla_d^\eps
\nabla_J^\eps u_\eps \right) \left( \left[ \nabla^\eps_a,
\nabla^\eps_c \right] \nabla_I^\eps u_\eps \right) \right|
\\
&\le \left\Vert \xi_{\eps}^a \nabla_c^\eps \nabla_I^\eps u_\eps
\right\Vert_{e_{\eps}} \cdot \left\Vert \vphantom{|^|} \left[
\nabla^\eps_a, \nabla^\eps_c \right] \nabla_I^\eps u_\eps
\right\Vert_{e_{\eps}}
\\
&= \left\Vert \xi_{\eps} \right\Vert_{e_{\eps}} \cdot \vert
\nabla_{\eps}^{(k)} u_\eps \vert \cdot \left\Vert \vphantom{|^|}
\left[ \nabla^\eps_a, \nabla^\eps_c \right] \nabla_I^\eps u_\eps
\right\Vert_{e_{\eps}},
\end{align*}
where the equality in the first line follows from skew-symmetry in
$a$, $c$, and the inequality on the third line follows from the
Cauchy-Schwarz inequality. Moreover, from Condition~(A), we have the
following estimates for the curvature on compact sets
\begin{equation}
\label{term7x} \|\nabla_{a_1}^\eps\dots\nabla_{a_l}^\eps
R^{\eps}_{ab}{}^c{}_d \|_{e_\eps}\leq \frac{F_l}{\eps^{2+l}}, \qquad
l \geq 0.
\end{equation}
Using this estimate with $l=0$ and the Ricci identity, we deduce the
existence of a combinatorial constant $N_k$ depending only on $k$
such that
\[
\left\Vert \vphantom{|^|} \left[ \nabla^\eps_a, \nabla^\eps_c
\right] \nabla_I^\eps u_\eps \right\Vert_{e_{\eps}} \le N_k
\frac{F_0}{\eps^2} \vert \nabla_{\eps}^{(k-1)} u_\eps \vert.
\]
Hence, we have
\begin{equation}
|I_2| \le N_k \frac{F_0}{\eps^2} \left\Vert \xi_{\eps}
\right\Vert_{e_{\eps}} \cdot \vert \nabla_{\eps}^{(k)} u_\eps \vert
\cdot \vert \nabla_{\eps}^{(k-1)} u_\eps \vert \le \frac{N_k F_0
M_0}{2} \left( \vert \nabla_{\eps}^{(k)} u_\eps \vert^2 +
\frac{1}{\eps^4} \vert \nabla_{\eps}^{(k-1)} u_\eps \vert^2 \right).
\label{ro-est2}
\end{equation}
on compact sets.

\vskip .2cm
Finally, we consider the term $I_3$. We then have, by the
Cauchy-Schwarz inequality
\begin{align*}
|I_3| = \left| e^{IJ}_\eps \left( \xi_{\eps}^d \nabla_d^\eps
\nabla_J^\eps u_\eps \right) \left( g^{ac}_\eps \nabla_a^\eps
\nabla_c^\eps \nabla_I^\eps u_\eps \right) \right| &\le \left\Vert
\xi_{\eps}^d \nabla_d^\eps \nabla_I^\eps u_\eps
\right\Vert_{e_{\eps}} \cdot \left\Vert g^{ac}_\eps \nabla_a^\eps
\nabla_c^\eps \nabla_I^\eps u_\eps \right\Vert_{e_{\eps}}
\\
&\le P_k \left\Vert \xi_{\eps} \right\Vert_{e_{\eps}} \cdot \vert
\nabla_{\eps}^{(k)} u_\eps \vert \cdot \left\Vert g^{ac}_\eps
\nabla_a^\eps \nabla_c^\eps \nabla_I^\eps u_\eps
\right\Vert_{e_{\eps}},
\end{align*}
where $P_k$ is a combinatorial constant depending only on $k$. Again
using the Ricci identities, and the fact that $u_\eps$ is a solution
of~\eqref{weqofsettingeps}, we may write
\[
g^{ac}_\eps \nabla_a^\eps \nabla_c^\eps \nabla_I^\eps \, u_\eps =
\nabla^\eps_{I}\, f_\eps + \sum_{j=1}^{k-1} \left(
\mathcal{R}_\eps^{(k-1, j)} u_\eps \right)_I,
\]
where $\mathcal{R}^{(k-1, j)}_\eps u_\eps$ denotes a linear
combination of contractions of the $(k-j-1)$'th covariant derivative
of the Riemann tensor with the $j$'th covariant derivative of
$u_\eps$. A second appeal to~\eqref{term7x} implies that on each
compact set there exists a constant $G_k$ such that
\[
\left\Vert \vphantom{|^|} \left( \mathcal R^{(k-1, j)}_\eps u_\eps
\right)_I \right\Vert_{e_\eps} \leq
\frac{G_k}{\eps^{k-j+1}}|\nabla_\eps^{(j)} u_\eps|, \qquad (\eps
\rightarrow 0).
\]
We therefore have
\begin{align}
|I_3| &\le P_k \left\Vert \xi_{\eps} \right\Vert_{e_{\eps}} \cdot
\vert \nabla_{\eps}^{(k)} u_\eps \vert \left( |\nabla_\eps^{(k-1)}
f_{\eps}| + G_k \sum_{j = 1}^{k-1} \frac{1}{\eps^{1+k-j}}
|\nabla_\eps^{(j)}u_{\eps}| \right)
\nonumber\\
&\le \frac{P_k M_0}{2} \cdot \left( k \vert \nabla_{\eps}^{(k)}
u_\eps \vert^2 + |\nabla_\eps^{(k-1)} f_{\eps}|^2 + G_k^2 \sum_{j =
1}^{k-1} \frac{1}{\eps^{2(1+k-j)}} |\nabla_\eps^{(j)}u_{\eps}|^2
\right) \label{ro-est3}
\end{align}

Putting together~\eqref{ro-est1}, \eqref{ro-est2},
and~\eqref{ro-est3}, we have
\[
\left| \xi_b \nabla_a^\eps T^{ab, k}_\eps(u_\eps) \right| \leq
\alpha_k |\nabla_\eps^{(k)} u_{\eps}|^2 + \beta_k
|\nabla_\eps^{(k-1)} f_{\eps} |^2 + \gamma_k \sum_{j = 1}^{k-1}
\frac{|\nabla_\eps^{(j)} u_{\eps}|^2}{\eps^{2(1+k-j)}},
\]
for positive constants $\alpha_k, \beta_k, \gamma_k$. Summation over
$k=1 \dots m$ and integration yields positive constants
$\widetilde{\alpha}_m$, $\widetilde{\beta}_m$,
$\widetilde{\gamma}_m$ such that
\[
\left| \sum_{k=0}^m \int_{\Omega_\tau} \xi_b \nabla_a^\eps T^{ab,
k}_\eps(u_{\eps}) \mu_\eps \right| \le \widetilde{\alpha}_m
(\SobODt{u_\eps}{m})^2 + \widetilde{\beta}_m
(\SobODt{f_\eps}{m-1})^2 + \widetilde{\gamma}_m \sum_{j=1}^
{m-1}\frac{1}{\eps^{2(1+m-j)} } (\SobODt{u_\eps}{j})^2.
\]
On substituting this inequality and~\eqref{term2} into
equation~\eqref{energyhierarchystokes}, we deduce that
\begin{eqnarray}
\label{lastmanstanding} E^m_{\tau, \eps}(u_\eps) &\leq & E^m_{0,
\eps}(u_\eps)\\ \nonumber &&+ \left( \widetilde{\alpha}_m +
\frac{1}{2} \tilde{A}_m K \right) (\SobODt{u_\eps}{m})^2 +
\widetilde{\beta}_m (\SobODt{f_\eps}{m-1})^2 + \widetilde{\gamma}_m
\sum_{j=1}^ {m-1}\frac{1}{\eps^{2(1+m-j)} } (\SobODt{u_\eps}{j})^2.
\end{eqnarray}
As in the case with $k=1$, we may use Lemma~\ref{lemma1} to write
\[
(\SobODt{u_\eps}{j})^2= \int_{\zeta=0}^\tau (\SobSDz{u_\eps}{j})^2\,
d\zeta\leq \frac{1}{A'} \int_{\zeta=0}^\tau E^j_{\zeta,
\eps}(u_\eps)d \zeta,
\]
for $j = 1, \dots, m$. Substituting these relations
into~\eqref{lastmanstanding} yields the
inequality~\eqref{energyinequalitylevelkformula}, with $C_m^{\prime}
:= \widetilde{\beta}_m$, $C_m^{\prime\prime} := \widetilde{\gamma}_m
/ A'$ and $C_m^{\prime\prime\prime} := ( \widetilde{\alpha}_m +
\frac{1}{2} \tilde{A}_m K ) / A'$.
\end{proof}

\begin{remark}
As can be seen from the expression for $\nabla_a^{\eps} T^{ab,
0}_{\eps}(u_{\eps})$, there is no estimate of the
form~\eqref{energyinequalitylevelkformula} for $E^0_{\tau,
\eps}(u_{\eps})$. However, $E^0_{\tau, \eps}(u_{\eps})$ is estimated
in terms of $E^k_{\tau, \eps}(u_{\eps})$, with $k \ge 1$; a fact
that is implicit in Proposition~\ref{energyinequality}.
\end{remark}

\section{Auxiliary estimates}
\label{auxest}

In this section, we complement the energy inequalities derived in
Section~\ref{partb} with estimates that allow us to utilise the
former in the proof of the main result. In particular, we shall
prove that
\begin{itemize}
\item[(i)] suitable bounds on the initial data give suitable bounds on the initial energies $E^k_{0, \eps}(u_\eps)$;
\item[(ii)] suitable bounds on the energies $E^k_{\tau, \eps}(u_\eps)$ give suitable bounds on the solution $u_\eps$.
\end{itemize}

The existence as well as the uniqueness part of the proof of the
main theorem will then use (i) combined with
Corollary~\ref{energiesviainitialenergies} and (ii) to establish
moderateness resp.\ negligibility of the candidate solution.

\begin{lemma}[Bounds on initial energies from initial data]
\label{initialenergiesviainitialdata} Let $u_\eps$ be a solution
of~\eqref{weqofsettingeps}. If $(v_\eps)_\eps$, $(w_\eps)_\eps$ are
moderate resp.\ negligible, then the initial energies $(E^k_{0,
\eps}(u_\eps))_\eps$, for each $k \ge 0$, are moderate resp.\
negligible nets of real numbers.
\end{lemma}

\begin{proof}
The estimates for the spatial derivatives $\pa_{x^{i_1}} \dots
\pa_{x^{i_k}} u_\eps(0, x^i) = \pa_{x^{i_1}} \dots \pa_{x^{i_k}}
v_\eps(x^i)$ are immediate. To estimate $\pa_t \pa_{x^{i_1}} \dots
\pa_{x^{i_k}} u_\eps(0, x^i)$, we rewrite the initial conditions in
equation~\eqref{weqofsettingeps} in the form
\begin{align*}
u_\eps(t=0, x^i) &= v_\eps(x^i)
\\
\nonumber \pa_t u_\eps (t=0, x^i) &= \tilde{w}_\eps(x^i),
\end{align*}
where we define $\tilde{w}_{\eps} := V_{\eps} w_{\eps} - N_{\eps}^i
\partial_{x^i} v_{\eps}$. It is straightforward to show, using the
asymptotic estimates~\eqref{coeffestimates}, that $(v_{\eps},
w_{\eps})$ being moderate resp.\ negligible implies moderateness
resp.\ negligibility of $(v_{\eps}, \tilde{w}_{\eps})$. Therefore
moderateness resp.\ negligibility of $(v_{\eps}, w_{\eps})$ implies
moderateness resp.\ negligibility of $\pa_t \pa_{x^{i_1}} \dots
\pa_{x^{i_k}} u_\eps(0, x^i) \equiv \pa_{x^{i_1}} \dots
\pa_{x^{i_k}} \tilde{w}_\eps(x^i)$.

The estimates for higher (mixed) time derivatives follow inductively
by rewriting the wave equation in the form
\[
\partial_t^2 u_\eps = -V_\eps^2 \left ( f_\eps + \frac{2}{V_{\eps}^2} N^i \partial_t \partial_i u_{\eps} - \left( h_{\eps}^{ij} - \frac{1}{V_{\eps}^2} N_{\eps}^i N_{\eps}^j \right) \partial_i \partial_j u_{\eps} + g_{\eps}^{ab} \, \Gamma[\mathbf{g}_{\eps}]^c{}_{ab} \, \frac{\partial u_{\eps}}{\partial x^c} \right)
\]
and using again the estimates~\eqref{coeffestimates} for $V_\eps$,
$N_{\eps}^i$, $h_{\eps}^{ij}$ as well as $f_\eps$, $v_{\eps}$,
$w_{\eps}$.
\end{proof}

\begin{lemma}[Bounds on solutions from bounds on energies]\label{ro-sobolev}
For $m > 3/2$ an integer, there exists a constant $K$ and number $N$
such that for all $u \in \Cc^{ \infty}(\Omega_\tau)$ and for all
$\zeta \in [0, \tau]$ we have
\[
\sup_{x \in \Omega_\tau} \left| \pa_{x^{a_1}} \cdots \pa_{x^{a_l}}
u(x) \right| \leq K \eps^{-N} \sup_{0 \leq \zeta \leq \tau}
E_{\zeta, \eps}^{m+l}(u).
\]
\end{lemma}
\begin{remark}
Note that the statement is for all $u  \in \Cc^{
\infty}(\Omega_\tau)$. In the proof of the main theorem, we will
apply it to a solution, $u_{\eps}$, of the wave equation.
\end{remark}

\begin{proof}[Proof of~\ref{ro-sobolev}]
First we combine the standard Sobolev embedding theorem on $S_\tau$
with the fact that by assumption (A) the metric and hence the volume
is $O(1)$ to obtain for $m > 3/2$
\begin{equation}\label{ro-emb}
\sup_{x \in S_\zeta} \vert u(x) \vert \leq K \ \SobSdz{u}{m}.
\end{equation}
Then we successively apply~\eqref{ineqSdXSD} and~\eqref{ineqEXSD} to
obtain
\[
\sup_{x \in S_\zeta} \vert u(x) \vert \leq \eps^{-N} E^m_{\zeta,
\eps}(u).
\]
Taking the supremum over $\zeta \in [0, \tau]$ on the right hand
side gives the result for $l=0$. To prove the general result, we
replace $u$ by the respective derivatives. In some more detail, note
that time derivatives are not covered by the Sobolev embedding
theorem since they are transversal to $S_\tau$, i.e., we have to
replace~\eqref{ro-emb} by the estimate
\[
\sup_{x \in S_\zeta} \vert \pa_{\rho_1} \dots \pa_{\rho_k} \pa_t^s u
\vert \leq K\ \SobSdz{\pa_t^s u}{m+k} \leq K\ \SobSdz{u}{m+k+s},
\]
where the last inequality holds because the norm $\SobSdz{\ }{m}$,
in addition, contains time derivatives.
\end{proof}

\section{Proof of the main theorem}
\label{proof}

We finally prove the main result by putting together the estimates
achieved so far.

\begin{proof}[Proof of~\ref{mainthm}]
{\ }\\[2mm]
\emph{Step~1: Existence of classical solutions.} Due to assumption
(C), classical theory provided us with smooth solutions for fixed
$\eps$. More precisely, by~\cite[Theorem 5.3.2]{FL1}, for $\eps$
fixed there exists a unique smooth function $u_\eps$
solving~\eqref{weqofsettingeps} on $A \subseteq \bigcap_{\eps <
\eps_0} J_\eps^+(\Sigma)$. Without loss of generality, we may assume
that $\Omega_{\gamma} \subseteq A$. 

\vskip .2cm
\noindent{}\emph{Step~2: Existence of $\mathcal{G}$-solutions
(moderateness estimates).} We show that the net $(u_\eps)_\eps$ of
Step~1 is moderate on $\Omega_{\gamma}$: Moderate data
$(v_\eps)_\eps$, $(w_\eps)_\eps$ translate, by means of
Lemma~\ref{initialenergiesviainitialdata}, to moderate initial
energies $(E^k_{0, \eps}(u_\eps))_\eps$ for each $k\geq1$. Moreover,
by means of Corollary~\ref{energiesviainitialenergies}, moderate
initial energies translate to moderate energies $(E^k_{\tau,
\eps}(u_\eps))_\eps$ ($k\geq 1$) for all $0\leq \tau\leq \gamma$.
Finally, it follows from Lemma~\ref{ro-sobolev} that moderate
energies $(E^k_{\tau, \eps}(u_\eps))_\eps$ ($k \geq 1, \, 0 \leq
\tau \leq \gamma$) imply moderateness of $(u_\eps)_\eps$. Hence
$u:=[(u_\eps)_\eps]$ is a generalised solution on $\Omega_\tau$ of
the i.v.p.~\eqref{weqofsettingeps}.

\vskip .2cm
\noindent{}\emph{Step~3: Uniqueness of $\mathcal{G}$-solutions
(negligibility estimates).} We are left with showing that the
solution $u$ does not depend on the choice of representatives of
$(f_\eps)_\eps$, $(v_\eps)_\eps$, $(w_\eps)_\eps$, and
$(\mathbf{g}_{\eps})_\eps$ of $f=0$, $v$, $w$, and $\mathbf{g}$.
Leaving the latter for Step~4, we observe that, to show independence
of the choice of representatives of $f$, $v$, and $w$, it suffices
to prove that if $(v_\eps)_\eps$ and $(w_\eps)_\eps$ are negligible,
then the corresponding solution $(u_\eps)_\eps$ is also negligible.
To establish this claim we argue as in Step~2 but using the
negligibility parts of Lemma~\ref{initialenergiesviainitialdata} and
Corollary~\ref{energiesviainitialenergies}. We then observe that
negligibility of the energies in Lemma~\ref{ro-sobolev} implies
negligibility of $(u_\eps)_\eps$.

\vskip .2cm
\noindent{}\emph{Step~4: Independence of the representative of
the metric.} We finally prove independence of the solution on
representatives $(\mathbf{g}_{\eps})_\eps$ of the metric. So let
$(\widehat{\mathbf{g}}_{\eps})_\eps$ be another representative of
$\mathbf{g}$. Denoting the corresponding d'Alembertian by
$\widehat{\Box}^\eps$ we consider the initial value problem
\begin{eqnarray}
\widehat{\Box}^\eps \hat{u}_\eps &=& f_\eps,
\nonumber\\
\hat{u}_\eps(t=0, x^i) &=& v_\eps(x^i),
\label{weqhat}\\
\partial_t \hat{u}_\eps(t=0, x^i) &=& w_\eps(x^i).
\nonumber
\end{eqnarray}
By Step~2, there exists a moderate net of solutions
$(\hat{u}_\eps)_\eps$ of~\eqref{weqhat}, and we only have to show
that its difference with the unperturbed solution,
$(\widetilde{u}_\eps)_\eps := (u_\eps)_\eps - (\hat{u}_\eps)_\eps$,
is negligible on $\Omega_\tau$. This difference is a solution of the
i.v.p.\
\begin{eqnarray}
\widehat{\Box}^\eps \widetilde{u}_\eps &=& f_\eps -
\widehat{\Box}^\eps u_\eps
\nonumber\\
\widetilde{u}_\eps(t=0, x^i) &=& 0
\label{weqhatdiff}\\
\partial_t \widetilde{u}_\eps(t=0, x^i)&=&0.
\nonumber
\end{eqnarray}
In view of Step~3, we only have to show that $f_\eps -
\widehat{\Box}^\eps u_\eps$ is negligible. To this end, we write
\[
f_\eps - \widehat{\Box}^\eps u_\eps = (f_\eps - \Box^\eps u_\eps) +
(\Box^\eps u_\eps - \hat \Box^\eps u_\eps) = \Box^\eps u_\eps -
\widehat{\Box}^\eps u_\eps,
\]
where we have used the fact that $(u_\eps)_\eps$
solves~\eqref{weqofsettingeps}. Therefore, the problem is reduced to
showing that $(\Box^\eps u_\eps - \widehat{\Box}^\eps u_\eps)_\eps$
is negligible. This, however, is clear since $\Box$ is a
well-defined differential operator on $\mathcal{G}$.
\end{proof}

\section{Conclusion}
We have proved unique local solvability of the wave equation for a
large class of metrics of low regularity in the Colombeau algebra of
generalised functions, hence establishing
$\mathcal{G}$-hyperbolicity of these space-times in the sense of
Vickers and Wilson~\cite{VW}. (This, in itself, is a slight
modification of Clarke's notion of generalised
hyperbolicity~\cite{Clarke}, in the sense that we now consider
solvability in $\mathcal{G}$ rather than $\mathscr{D}'$.) The
essential assumption on this class of metrics is local boundedness:
in particular, it includes conical space-times, and therefore
generalises the results of Vickers and Wilson~\cite{VW}. Our class
of metrics also includes non-static examples such as impulsive
$pp$-waves and expanding spherical impulsive waves.

Finally, we remark that the regularity assumptions~(A) and~(B) on
the metric may be relaxed slightly. Indeed, we can replace the
$O(1)$-asymptotics for the zeroth order derivative of the metric in
Condition~(A) as well as in Condition~(B) by the condition that
these quantities be $O(\log(1/\eps))$. (This corresponds to
generalised H\"{o}lder-Zygmund regularity of order zero of the
metric as defined in \cite{GH}.) Under these conditions, the
constants $A$ and $A'^{-1}$ in \eqref{ineqEXSD} of
Lemma~\ref{lemma1} as well as of $C_k'''$ in
Proposition~\ref{energyinequality} and
Corollary~\ref{applicationenergygronwall} have a growth behaviour of
$O(\log(1/\eps))$. However, Corollary
\ref{energiesviainitialenergies} remains unchanged since the
$O(\log(1/\eps))$-growth together with Gronwall's lemma still yield
moderateness resp.\ negligibility estimates. Therefore, given a
classical metric which we regularise (either by convolution with a
mollifier or by any physically motivated procedure) subject to these
weaker asymptotic conditions, then our existence and uniqueness result still holds.

{}From these considerations, we also see that it
is hard to imagine how the regularity assumptions for the metric
could be further relaxed within our framework.

\end{document}